\documentclass[11pt,letterpaper]{article}  % Comment this line out
\usepackage{amsthm} % assumes amsmath package installed
\usepackage{amssymb}  % assumes amsmath package installed

\newtheorem{theorem}{Theorem}

\newtheorem{lemma}{Lemma}
\newtheorem{remark}{Remark}
\newtheorem{definition}{Definition}

\usepackage[margin=1in]{geometry}                                                    
\usepackage{amsmath,amsfonts,amssymb,color}
\usepackage{epsfig}
\usepackage{algorithm}
\usepackage{algorithmic}
\usepackage{multirow}

\usepackage{array}
\usepackage{multirow}
\usepackage{epstopdf}
\usepackage{tikz}
\usepackage{relsize}
\usetikzlibrary{shapes,arrows}
\usepackage{cite}
\usepackage{pmat}
\usepackage{epstopdf}
\usepackage{tikz}
\usetikzlibrary{shapes}
\usepackage{scalerel}
\definecolor{winered}{rgb}{0.5,0,0}
\usepackage{scalerel}
\usepackage{xcolor}
\usepackage[colorlinks]{hyperref}
\AtBeginDocument{%
  \hypersetup{
    citecolor={winered},
    linkcolor=blue!}}

\title{\LARGE \bf A New Approach for Distributed Hypothesis Testing with Extensions to Byzantine-Resilience}
\author{Aritra Mitra, John A. Richards and Shreyas Sundaram
\thanks{A. Mitra and S. Sundaram are with the School of Electrical and Computer Engineering at Purdue University. J. A.  Richards is with Sandia National Laboratories.   Email: {\tt \{mitra14, sundara2\}@purdue.edu},  {\tt{jaricha@sandia.gov}}. This work was supported in part by NSF CAREER award
1653648, and by a grant from Sandia National Laboratories. Sandia National Laboratories is a multimission laboratory managed and operated by National Technology \& Engineering Solutions of Sandia, LLC, a wholly owned subsidiary of Honeywell International Inc., for the U.S. Department of Energy's National Nuclear Security Administration under contract DE-NA0003525. The views expressed in the article do not necessarily represent the views of the U.S. Department of Energy or the United States Government.}}
\begin{document}
\maketitle
\begin{abstract}
We study a setting where a group of agents, each receiving partially informative private observations, seek to collaboratively learn the true state (among a set of hypotheses) that explains their joint observation profiles over time. To solve this problem, we propose a distributed learning rule that differs fundamentally from existing approaches, in the sense, that it does not employ any form of ``belief-averaging''. Specifically, every agent maintains a local belief (on each hypothesis) that is updated in a Bayesian manner without any network influence, and an actual belief that is updated (up to normalization) as the minimum of its own local belief and the actual beliefs of its neighbors. Under minimal requirements on the signal structures of the agents and the underlying communication graph, we establish consistency of the proposed belief update rule, i.e., we show that the actual beliefs of the agents asymptotically concentrate on the true state almost surely. As one of the key benefits of our approach, we show that our learning rule can be extended to scenarios that capture misbehavior on the part of certain agents in the network, modeled via the Byzantine adversary model. In particular, we prove that each non-adversarial agent can asymptotically learn the true state of the world almost surely, under appropriate conditions on the observation model and the network topology. 
\end{abstract}
\section{Introduction}
Various distributed learning problems arising in social networks (such as opinion formation and spreading), and in engineering systems (such as target recognition by a group of aerial robots) can be studied under the formal framework of distributed hypothesis testing. Within this framework, a group of agents repeatedly observe certain private signals, and aim to infer the ``true state of the world" that explains their joint observations. While much of the earlier work on this topic assumed the existence of a centralized fusion center for performing computational tasks \cite{fusion1,fusion2}, more recent endeavors focus on a distributed setting where interactions among agents are captured by a communication graph \cite{GEBjad,jad2,liu,nedic1,nedic2,lalitha1,su1,shahinTAC,rad,shahinparam}. Our work here falls in the latter class. A typical belief update rule in the distributed setting combines a local Bayesian update with a consensus-based opinion pooling of neighboring beliefs. Specifically, linear opinion pooling is studied in \cite{GEBjad,jad2,liu}, whereas the log-linear form of belief aggregation is studied in the context of distributed hypothesis testing in \cite{nedic1,nedic2,lalitha1,su1,shahinTAC}, and distributed parameter estimation in \cite{rad,shahinparam}. Notably, exponential convergence rates are achieved in \cite{jad2,nedic1,nedic2,lalitha1,su1}, while a finite-time analysis is presented in \cite{shahinTAC}. Extensions to time-varying graphs have also been studied in \cite{liu,nedic1,nedic2}.

In \cite[Section III]{nedic2}, the authors explain that the commonly studied linear and log-linear forms of belief aggregation are specific instances of a more general class of opinion pooling known as g-Quasi-Linear Opinion pools (g-QLOP), introduced in \cite{pools}. The main contribution of our paper is the development of a novel belief update rule that deviates fundamentally from the broad family of g-QLOP learning rules discussed above. Specifically, the learning algorithm that we propose in Section \ref{sec:Algo} does not rely on any linear consensus-based belief aggregation protocol. Instead, each agent maintains two sets of beliefs: a local belief that is updated in a Bayesian manner based on the private observations (without neighbor interactions), and an actual belief that is updated (up to normalization) as the minimum of the agent's own local belief and the actual beliefs of its neighbors. In Section \ref{sec:Proofs}, we establish that under minimal requirements on the agents' signal structures and the communication graph, the actual beliefs of the agents asymptotically concentrate on the true state almost surely. In Section \ref{sec:discuss}, we argue that our approach works under graph-theoretic conditions that are milder than the standard assumption of strong-connectivity.

In addition to the above contribution to the distributed hypothesis testing problem, we also show in this paper that our approach is capable of handling agents that do not follow the prescribed learning algorithm. Indeed, despite the wealth of literature on distributed inference, there is limited understanding of the impact of misbehaving agents for the problem under consideration. Such agents may represent stubborn individuals, ideological extremists in the context of a social network, or model faults (either benign or malicious) in a networked control system. \textit{In the presence of such misbehaving entities, how should the remaining agents process their private observations and the beliefs of their neighbors to eventually learn the truth?}  To answer this question, we model misbehaving agents via the classical Byzantine adversary model, and develop a provably correct, resilient version of our proposed learning rule in Section \ref{sec:LFRHE}. The only related work (that we are aware of) in this regard is reported in \cite{su1}. As we discuss in Section \ref{sec:LFRHE}, our proposed approach is significantly less computationally intensive relative to those in \cite{su1}. We identify conditions on the observation model and the network structure that guarantee applicability of our Byzantine-resilient learning rule, and argue (in Section \ref{sec:discuss}) that such conditions can be checked in polynomial time.
\section{Model and Problem Formulation}
\label{sec:model}
\textbf{Network Model:} We consider a group of  agents $\mathcal{V}=\{1,2,\ldots,n\}$ interacting over a time-invariant, directed communication graph $\mathcal{G}=(\mathcal{V},\mathcal{E})$. An edge $(i,j)\in\mathcal{E}$ indicates that agent $i$ can directly transmit information to agent $j$. If $(i,j)\in\mathcal{E}$, then agent $i$ will be called a neighbor of agent $j$, and agent $j$ will be called an out-neighbor of agent $i$. The set of all neighbors of agent $i$ will be denoted $\mathcal{N}_i$. Given two disjoint sets $\mathcal{C}_1,\mathcal{C}_2 \subseteq{\mathcal{V}}$, we say that $\mathcal{C}_2$ is reachable from $\mathcal{C}_1$ if for every $i\in\mathcal{C}_2$, there exists a directed path from some $j\in\mathcal{C}_1$ to agent $i$ (note that $j$ will in general be a function of $i$). We will use $|\mathcal{C}|$ to denote the cardinality of a set $\mathcal{C}$.

\textbf{Observation Model:} Let $\Theta=\{\theta_1,\theta_2,\ldots,\theta_m\}$ denote $m$ possible states of the world; each $\theta_i\in\Theta$ will be called a hypothesis. Let  $\mathbb{N}$ and $\mathbb{N}_{+}$ denote the set of non-negative integers and positive integers, respectively. Then at each time-step $t\in\mathbb{N}_{+}$, every agent $i\in\mathcal{V}$  privately observes a signal $s_{i,t}\in\mathcal{S}_i$, where $\mathcal{S}_i$ denotes the signal space of agent $i$. The joint observation profile so generated across the network is denoted ${s}_{t}=(s_{1,t},s_{2,t},\ldots,s_{n,t})$, where $s_t\in\mathcal{S}$, and $\mathcal{S}=\mathcal{S}_1\times\mathcal{S}_2\times\ldots \mathcal{S}_n$. 
The signal $s_{t}$ is generated based on a conditional likelihood function $l(\cdot|\theta^{\star})$, governed by the true state of the world $\theta^{\star}\in\Theta$. Let $l_i(\cdot|\theta^{\star}), i\in\mathcal{V}$ denote the $i$-th marginal of $l(\cdot|\theta^{\star})$. 
The signal structure of each agent $i\in\mathcal{V}$ is then characterized by a family of parameterized marginals $\{l_i(w_i|\theta): \theta\in\Theta, w_i\in\mathcal{S}_i\}$.\footnote{Whereas $w_i\in\mathcal{S}_i$ will be used to refer to a generic element of the signal space of agent $i$,  $s_{i,t}$  will denote the random variable (with distribution $l_i(\cdot|\theta^{\star})$) that corresponds to agent $i'$s observation at time-step $t$.}

We make the following standard assumptions \cite{GEBjad,jad2,liu,nedic1,nedic2,lalitha1,su1,shahinTAC}: (i) The signal space of each agent $i$, namely $\mathcal{S}_i$, is finite. (ii) Each agent $i$ has knowledge of its local likelihood functions $\{l_i(\cdot|\theta_p)\}_{p=1}^{m}$, and it holds that $l_i(w_i|\theta) > 0, \forall w_i\in\mathcal{S}_i$, and $\forall \theta \in \Theta$. (iii) The observation sequence of each agent is described by an i.i.d. random process over time; however, at any given time-step, the observations of different agents may potentially be correlated. (iv) There exists a fixed true state of the world $\theta^{\star}\in\Theta$ (unknown to the agents) that generates the observations of all the agents.\footnote{The approach in \cite{nedic1,nedic2} applies to a more general setting where there may not exist such a true hypothesis.} Finally, we define a probability triple $(\Omega,\mathcal{F},\mathbb{P}^{\theta^{\star}})$, where $\Omega\triangleq\{\omega: \omega=(s_1,s_2,\ldots), \forall s_t\in\mathcal{S}, \forall t \in \mathbb{N}_{+}\}$, $\mathcal{F}$ is the $\sigma$-algebra generated by the observation profiles, and $\mathbb{P}^{\theta^{\star}}$ is the probability measure induced by sample paths in $\Omega$. Specifically, $\mathbb{P}^{\theta^{\star}}=\prod \limits_{t=1}^{\infty}l(\cdot|\theta^{\star})$. For the sake of brevity, we will say that an event occurs almost surely to mean that it occurs almost surely w.r.t. the probability measure $\mathbb{P}^{\theta^{\star}}$.

Given the above setup, the goal of each agent in the network is to discern the true state of the world $\theta^{\star}$. The challenge associated with such a task stems from the fact that the private signal structure of any given agent is in general only partially informative. To make this notion precise, define $\Theta^{\theta^{\star}}_i\triangleq\{\theta\in\Theta : l_i(w_i|\theta)=l_i(w_i|\theta^{\star}), \forall w_i\in\mathcal{S}_i\}.$ In words, $\Theta^{\theta^{\star}}_i$ represents the set of hypotheses that are \textit{observationally equivalent} to the true state $\theta^{\star}$ from the perspective of agent $i$. In general, for any agent $i\in\mathcal{V}$,  we may have $|\Theta^{\theta^{\star}}_i| > 1$, necessitating collaboration among agents. While inter-agent collaboration is implicitly assumed in the distributed hypothesis testing literature, in this paper we will also allow for potential misbehavior on the part of certain agents in the network, modeled as follows.

\textbf{Adversary Model:} We assume that a certain fraction of the agents are adversarial, and model their behavior based on the Byzantine fault model \cite{Byz}. In particular, Byzantine agents possess complete knowledge of the observation model, the network model, the algorithms being used, the information being exchanged, and the true state of the world. Leveraging such information, adversarial agents can behave arbitrarily and in a coordinated manner, and can in particular, send incorrect, potentially inconsistent information to their out-neighbors. In terms of their distribution in the network, we will consider an $f$-local adversarial model, i.e., we assume that there are at most $f$ adversaries in the neighborhood of any non-adversarial agent.\footnote{Note that the $f$-local adversarial model assumed here is more general than the $f$-total adversarial model considered in \cite{su1}, where the total number of adversaries in the entire network is upper bounded by $f$.} Finally, we emphasize that the non-adversarial agents are unaware of the identities of the adversaries in their neighborhood. As is fairly standard in the distributed fault-tolerant literature \cite{vaidyacons,rescons,suopt,sundaramopt,mitraarxiv,broad1}, we only assume that non-adversarial agents know the upper bound $f$ on the number of adversaries in their neighborhood. The adversarial set will be denoted by $\mathcal{A}\subset\mathcal{V}$, and the remaining agents $\mathcal{R}=\mathcal{V}\setminus\mathcal{A}$ will be called the regular agents.

Our \textbf{objective} in this paper will be to design a distributed learning rule that allows each regular agent $i\in\mathcal{R}$ to identify the true state of the world almost surely, despite (i) the partially informative signal structures of the agents, and (ii) the actions of any $f$-local Byzantine adversarial set. To this end, we introduce the following notion of \textit{source agents}.
\begin{definition}(\textbf{Source agents}) An agent $i$ is said to be a source agent for a pair of distinct hypotheses $\theta_p,\theta_q\in\Theta$, if $D(l_i(\cdot|\theta_p)||l_i(\cdot|\theta_q)) > 0$, where $D(l_i(\cdot|\theta_p)||l_i(\cdot|\theta_q))$ represents the KL-divergence between the distributions $l_i(\cdot|\theta_p)$ and $l_i(\cdot|\theta_q)$, and is given by:
\begin{equation}
D(l_i(\cdot|\theta_p)||l_i(\cdot|\theta_q))=\sum \limits_{w_i\in\mathcal{S}_i}l_i(w_i|\theta_p)\log\frac{l_i(w_i|\theta_p)}{l_i(w_i|\theta_q)}.
\end{equation}
The set of all source agents for the pair $\theta_p,\theta_q$ is denoted by $\mathcal{S}(\theta_p,\theta_q)$.\footnote{Notice that $\mathcal{S}(\theta_p,\theta_q)=\mathcal{S}(\theta_q,\theta_p)$, since $D(l_i(\cdot|\theta_p)||l_i(\cdot|\theta_q))>0 \iff D(l_i(\cdot|\theta_q)||l_i(\cdot|\theta_p))>0$.}
\end{definition}

In words, a source agent for a pair $\theta_p,\theta_q\in\Theta$ is an agent that can distinguish between the pair of hypotheses $\theta_p,\theta_q$ based on its private signal structure. In our developments, we will require the following two definitions.
\begin{definition}($r$-\textbf{reachable set}) \cite{rescons} For a graph $\mathcal{G}=(\mathcal{V,E})$, a set $\mathcal{C} \subseteq \mathcal{V}$, and an integer $r \in \mathbb{N}_{+}$, $\mathcal{C}$ is an \textit{$r$-reachable set} if there exists an $i \in \mathcal{C}$ such that $|\mathcal{N}_i \setminus \mathcal{C}| \geq r$.
\label{defn:rreachable}
\end{definition}
\begin{definition}(\textbf{strongly} $r$-\textbf{robust graph} \textit{w.r.t.} $\mathcal{S}(\theta_p,\theta_q)$) For $r \in \mathbb{N}_{+}$ and $\theta_p,\theta_q\in\Theta$, a graph $\mathcal{G}=(\mathcal{V,E})$ is \textit{strongly $r$-robust w.r.t. the set of source agents $\mathcal{S}(\theta_p,\theta_q)$}, if for every non-empty subset $\mathcal{C} \subseteq \mathcal{V}\setminus\mathcal{S}(\theta_p,\theta_q)$, $\mathcal{C}$ is $r$-reachable.
\label{defn:strongrobust}
\end{definition}
\section{Proposed Learning Rules}
\subsection{A Novel Belief Update Rule}
\label{sec:Algo}
In this section, we propose a novel belief update rule and discuss the intuition behind it. To introduce the key ideas underlying our basic approach, we first consider a scenario where all agents are regular, i.e., $\mathcal{R}=\mathcal{V}$. Every agent $i$ maintains and updates (at every time-step) two separate sets of belief vectors, namely, $\boldsymbol{\pi}_{i,t}$ and $\boldsymbol{\mu}_{i,t}$. Each of these vectors are probability distributions over the hypothesis set $\Theta$. We will refer to $\boldsymbol{\pi}_{i,t}$ and $\boldsymbol{\mu}_{i,t}$ as the ``local" belief vector (for reasons that will soon become obvious), and the ``actual" belief vector, respectively, maintained by agent $i$. The \textbf{goal} of each agent $i\in\mathcal{V}$ in the network will be to use its own private signals, and the information available from its neighbors, to update $\boldsymbol{\mu}_{i,t}$ sequentially so that $\lim_{t\to\infty}\mu_{i,t}(\theta^{*})=1$ almost surely. To do so, for each $\theta\in\Theta$, and at each time-step $t+1,t\in\mathbb{N}$, agent $i$ first generates $\pi_{i,t+1}(\theta)$ via a local Bayesian update rule that incorporates the private observation $s_{i,t+1}$ using $\pi_{i,t}(\theta)$ as a prior. Having generated $\pi_{i,t+1}(\theta)$, agent $i$ updates $\mu_{i,t+1}(\theta)$ (up to normalization) by setting it to be the minimum of its locally generated belief $\pi_{i,t+1}(\theta)$, and the actual beliefs $\mu_{j,t}(\theta), j\in \mathcal{N}_i$ of its neighbors at the previous time-step. It then reports its actual belief $\mu_{i,t+1}(\theta)$ to each of its out-neighbors.\footnote{Note that based on our algorithm, agents only exchange their actual beliefs, and not their local beliefs.} The belief vectors are initialized as $\mu_{i,0}(\theta)>0,\pi_{i,0}(\theta)>0, \forall \theta\in\Theta, \forall i\in\mathcal{V}$. Subsequently, these vectors are updated at each time-step $t+1$ (where $t\in\mathbb{N}$) as follows:
\begin{itemize}
\item \underline{\textbf{Step 1: Update of the local beliefs}:}
\begin{equation}
\pi_{i,t+1}(\theta)=\frac{l_i(s_{i,t+1}|\theta)\pi_{i,t}(\theta)}{\sum  \limits_{p=1}^{m} l_i(s_{i,t+1}|\theta_p)\pi_{i,t}(\theta_p)}.
\label{eqn:Bayes}
\end{equation}
\item \underline{\textbf{Step 2: Update of the actual beliefs}:}
\begin{equation}
\mu_{i,t+1}(\theta)=\frac{\min\{\{\mu_{j,t}(\theta)\}_{{j\in\mathcal{N}_i}},\pi_{i,t+1}(\theta)\}}{\sum\limits_{p=1}^{m}\min\{\{\mu_{j,t}(\theta_p)\}_{{j\in\mathcal{N}_i}},\pi_{i,t+1}(\theta_p)\}}.
\label{eqn:rule1}
\end{equation}
\end{itemize}

{\bf Intuition behind the learning rule}: Consider the set of source agents $\mathcal{S}(\theta^{*},\theta)$ who can differentiate between a certain false hypothesis $\theta$ and the true state $\theta^{\star}$. Suppose for now that this set is non-empty. We ask: how do the agents in the set $\mathcal{S}(\theta^{\star},\theta)$ contribute to the process of collaborative learning? To answer this question, we note that the signal structures of such agents are rich enough for them to be able to eliminate $\theta$ on their own, i.e., without the support of their neighbors. Thus, the agents in $\mathcal{S}(\theta^{\star},\theta)$ should contribute towards driving the actual beliefs of their out-neighbors (and eventually, of all the agents in the set $\mathcal{V}\setminus \mathcal{S}(\theta^{\star},\theta)$) on the hypothesis $\theta$ to zero. To achieve the above objective, we are especially interested in devising a rule that ensures that the capability of the source agents $\mathcal{S}(\theta^{\star},\theta)$ to eliminate $\theta$ is not diminished due to neighbor interactions. As we shall see later, such a property will be particularly useful when certain agents in the network are adversarial. It is precisely these considerations that motivate us to employ (i) an auxiliary belief vector $\boldsymbol{\pi}_{i,t+1}$ generated via local processing (i.e., without any network influence) of the private signals, and (ii) a min-rule of the form \eqref{eqn:rule1}. Specifically, if $i\in\mathcal{S}(\theta^{\star},\theta)$, then the sequence of local beliefs $\pi_{i,t+1}(\theta)$ will almost surely converge to $0$ based on the update rule \eqref{eqn:Bayes}. Hence, for a source agent $i\in \mathcal{S(\theta^{\star},\theta)}$, $\pi_{i,t+1}(\theta)$ will play the key role of an external network-independent input in the min-rule \eqref{eqn:rule1}. This in turn will trigger a process of belief reduction on the hypothesis $\theta$ originating at the source set $\mathcal{S(\theta^{\star},\theta)}$, and eventually propagating via the proposed min-rule to each agent in the network reachable from such a source set. The above discussion will be made precise in Section 
\ref{sec:Proofs}. 

\begin{remark}
We emphasize that the proposed min-rule \eqref{eqn:rule1} does not employ any form of ``belief-averaging''. This feature is in stark contrast with existing approaches to distributed hypothesis testing that rely either on linear opinion pooling \cite{GEBjad,jad2,liu}, or log-linear opinion pooling\cite{rad,shahinparam,shahinTAC,nedic1,nedic2,lalitha1,su1}. As such, the lack of linearity in our belief update rule precludes (direct or indirect) adaptation of existing analysis techniques to suit our needs. Consequently, we develop a novel sample path based proof technique in Section \ref{sec:Proofs} to establish consistency of the proposed learning rule. As one of the main outcomes of this analysis, we argue that our learning rule works under graph-theoretic conditions that are in general weaker than strong-connectivity (see also Section \ref{sec:discuss}).
\end{remark}
\subsection{A Byzantine-Resilient Belief Update Rule}
\label{sec:LFRHE}
As pointed out in the Introduction, a key benefit of our approach is that it can be extended to account for the worst-case Byzantine adversarial model described in Section \ref{sec:model}. A standard way to analyze the impact of such adversarial agents while designing resilient distributed consensus-based protocols (for applications in consensus \cite{rescons,vaidyacons}, optimization \cite{suopt,sundaramopt}, hypothesis testing \cite{su1},  and  multi-agent rendezvous \cite{seth1}) is to construct an equivalent matrix representation of the linear update rule that involves only the regular agents \cite{vaidyamatrix}. In particular, this requires expressing the iterates of a regular agent as a convex combination of the iterates of its regular neighbors, based on appropriate filtering techniques, and under certain assumptions on the network structure. While this can indeed be achieved efficiently for scalar consensus problems, for problems requiring consensus on vectors (like the belief vectors in our setting), such an approach becomes computationally prohibitive \cite{su1}. To bypass such heavy computations, and yet accommodate Byzantine attacks, we now develop a resilient version of the learning rule introduced in Section \ref{sec:Algo}, as follows. Each agent $i\in\mathcal{R}$ acts as follows at every time-step $t+1$ (where $t\in\mathbb{N}$).
\begin{itemize}
\item \underline{\textbf{Step 1: Update of the local beliefs}:} The local belief $\pi_{i,t+1}(\theta)$ is updated as before, based on \eqref{eqn:Bayes}. 
\vspace{3mm}
\item \underline{\textbf{Step 2: Filtering extreme beliefs}:} If $|\mathcal{N}_i| \geq (2f+1)$, then agent $i$ performs a filtering operation as follows. It collects the actual beliefs $\mu_{j,t}(\theta)$ from each neighbor $j\in\mathcal{N}_i$ and sorts them from highest to lowest. It rejects the highest $f$ and the lowest $f$ of such beliefs (i.e., it throws away $2f$ beliefs in all). In other words, for each hypothesis, a regular agent retains only the moderate beliefs received from its neighbors.
\vspace{3mm} 
\item \underline{\textbf{Step 3: Update of the actual beliefs}:} 
If $|\mathcal{N}_i| \geq (2f+1)$, then agent $i$ updates $\mu_{i,t+1}(\theta)$ as follows. Let the set of neighbors whose beliefs on $\theta$ are not rejected by agent $i$ (based on the previous filtering step) be denoted by $\mathcal{M}^{\theta}_{i,t}\subset\mathcal{N}_i$. The actual belief $\mu_{i,t+1}(\theta)$ is then updated as follows:
\begin{equation}
\mu_{i,t+1}(\theta)=\frac{\min\{\{\mu_{j,t}(\theta)\}_{j\in\mathcal{M}^{\theta}_{i,t}},\pi_{i,t+1}(\theta)\}}{\sum\limits_{p=1}^{m}\min\{\{\mu_{j,t}(\theta_p)\}_{j\in\mathcal{M}^{\theta_p}_{i,t}},\pi_{i,t+1}(\theta_p)\}}.
\label{eqn:rule2}
\end{equation}
If $|\mathcal{N}_i| < (2f+1)$, then agent $i$ updates $\mu_{i,t+1}(\theta)$ as follows:
\begin{equation}
\mu_{i,t+1}(\theta)=\pi_{i,t+1}(\theta).
\label{eqn:rule3}
\end{equation}
\end{itemize}
As with the learning rule presented in Section \ref{sec:Algo}, agent $i$ transmits $\mu_{i,t+1}(\theta)$ to each of its out-neighbors on completion of the above steps. We will refer to the above sequence of actions as the Local-Filtering based Resilient Hypothesis Elimination (LFRHE) algorithm. 
\section{Main Results}
In this section, we state our main results, and then comment on them in Section \ref{sec:discuss}; detailed proofs of the results are presented in Section \ref{sec:Proofs}.
Our first result establishes the correctness of the learning rule proposed in Section \ref{sec:Algo}.
\begin{theorem}
Suppose $\mathcal{R}=\mathcal{V}$, and that the following are true:
\begin{enumerate}
\item[(i)] For every pair of hypotheses $\theta_p,\theta_q\in\Theta$, the corresponding source set $\mathcal{S}(\theta_p,\theta_q)$ is non-empty.
\item[(ii)] For every pair of hypotheses $\theta_p,\theta_q\in\Theta$, $\mathcal{V}\setminus\mathcal{S}(\theta_p,\theta_q)$ is reachable from the source set $\mathcal{S}(\theta_p,\theta_q)$.
\item[(iii)] Every agent $i\in\mathcal{V}$ has a non-zero prior belief on each hypothesis, i.e., $\pi_{i,0}(\theta) > 0,\mu_{i,0}(\theta) > 0$ for all $i\in\mathcal{V}$, and for all $\theta\in\Theta$.
\end{enumerate}
Then, the learning rule described by equations \eqref{eqn:Bayes} and \eqref{eqn:rule1} leads to collaborative learning of the true state, i.e., $\mu_{i,t}(\theta^{\star}) \rightarrow 1$ almost surely $\forall i\in\mathcal{V}$.
\label{thm:rule1}
\end{theorem}

Our second result establishes the correctness of the LFRHE algorithm proposed in Section \ref{sec:LFRHE}.
\begin{theorem}
Suppose the following are true:
\begin{enumerate}
\item[(i)] For every pair of hypotheses $\theta_p,\theta_q\in\Theta$, the graph $\mathcal{G}$ is strongly $(2f+1)$-robust w.r.t. the corresponding source set $\mathcal{S}(\theta_p,\theta_q)$.
\item[(ii)] Each regular agent $i\in\mathcal{R}$ has a non-zero prior belief on each hypothesis, i.e., $\pi_{i,0}(\theta) > 0,\mu_{i,0}(\theta) > 0$ for all $i\in\mathcal{R}$, and for all $\theta\in\Theta$.
\end{enumerate}
Then, the LFRHE algorithm described by equations \eqref{eqn:Bayes}, \eqref{eqn:rule2} and \eqref{eqn:rule3} leads to collaborative learning of the true state despite the actions of any $f$-local set of Byzantine adversaries, i.e., $\mu_{i,t}(\theta^{\star}) \rightarrow 1$ almost surely $\forall i\in\mathcal{R}$.
\label{thm:rule2}
\end{theorem}
\begin{remark}
For any pair $\theta_p,\theta_q\in\Theta$, notice that condition (i) of Theorem \ref{thm:rule2} (together with the definition of strong-robustness in Def. \ref{defn:strongrobust}) requires $|\mathcal{S}(\theta_p,\theta_q)|\geq(2f+1)$,  if $\mathcal{V}\setminus\mathcal{S}(\theta_p,\theta_q)$ is non-empty.
\label{rem:redundancy}
\end{remark}
 \section{Discussion}
 \label{sec:discuss}
\textbf{(Assumptions in Theorem 1)}:
While the first condition in Theorem \ref{thm:rule1} is a basic global identifiability condition, the second condition on the network structure is in general weaker than the standard assumption of strong-connectivity made in \cite{GEBjad,jad2,rad,shahinparam,shahinTAC,lalitha1}. To see why the latter statement is true, consider a scenario where $\Theta=\{\theta_1,\theta_2\}$. Clearly, any agent $i\in\mathcal{S}(\theta_1,\theta_2)$ can discern the true state without neighbor interactions, precluding the need for incoming edges to such agents.\footnote{For the problem under consideration, the argument that the strong connectivity assumption can be relaxed applies to more general scenarios as well, where there does not necessarily exist any one agent that can identify the true state based on just its private signal structure. The underlying reason for this stems from information heterogeneity and information redundancy among agents \cite{mitraTAC}, features shared by distributed estimation and detection type problems, but lacking in a standard consensus setting. } Finally, the assumption of non-zero initial beliefs is fairly standard, and can be easily met by maintaining a uniform support over the hypotheses set initially.

\textbf{(Assumptions in Theorem 2)}: The first condition in Theorem \ref{thm:rule2} blends requirements on the signal structures of the agents with those on the communication graph. To gain intuition about this condition, suppose $\Theta=\{\theta_1,\theta_2\}$, and let there exist at least one agent $i\in\mathcal{V}\setminus\mathcal{S}(\theta_1,\theta_2)$. To enable agent $i$ to discern the truth despite potential adversaries in its neighborhood, one requires (i) redundancy in the signal structures of the agents (see Remark \ref{rem:redundancy}), and (ii) redundancy in the network structure to facilitate reliable information flow from $\mathcal{S}(\theta_1,\theta_2)$ to agent $i$. These requirements are captured by condition (i), a point made apparent in Section \ref{sec:proofthm2}.

\textbf{(Complexity of Checking Condition (i) in Theorem \ref{thm:rule2})}: Given a network of agents with associated signal structures, condition (i) in Theorem \ref{thm:rule2} can be checked in polynomial time. Specifically, for every pair $\theta_p,\theta_q\in\Theta$, finding the source set $\mathcal{S}(\theta_p,\theta_q)$ can easily be done in polynomial time via inspection of the agents' signal structures. For a fixed source set $\mathcal{S}(\theta_p,\theta_q)$, checking whether $\mathcal{G}$ is strongly $(2f+1)$-robust w.r.t. $\mathcal{S}(\theta_p,\theta_q)$ amounts to simulating a bootstrap percolation process on $\mathcal{G}$, with $\mathcal{S}(\theta_p,\theta_q)$ as the initial active set, and $(2f+1)$ as the threshold. This too can be achieved in polynomial time, as discussed in \cite{mitraarxiv}.

\textbf{(Analogy with Distributed State Estimation)}:
Consider the problem of collaboratively estimating the state of an LTI process based on information exchanges among agents that receive partial measurements of the state. There are natural connections between this setting, and the problem studied in this paper. For the state estimation scenario, one can fix an unstable mode of the process, and define source agents for that mode to be agents that can detect the eigenspaces associated with that mode. Interestingly, with source agents defined for each unstable mode in the manner described above, \cite[Theorem 3]{mitraTAC}  and \cite[Theorem 7]{mitraarxiv} (in the context of distributed state estimation) can be viewed as analogues of Theorem \ref{thm:rule1} and Theorem \ref{thm:rule2}, respectively. 

\textbf{(Convergence Rate)}: Consider any false hypothesis $\theta \neq \theta^{\star}$. We conjecture that based on our learning rules, the actual beliefs of all the regular agents on $\theta$ will almost surely decay exponentially fast after a transient period, with the rate of decay lower bounded by $\min_{i \in \mathcal{S}(\theta^{\star},\theta)\cap\mathcal{R}}D(l_i(\cdot|\theta^{\star})||l_i(\cdot|\theta))$.
\section{Proofs of the Main Results}
\label{sec:Proofs}
We start with the following simple lemma that characterizes the asymptotic behavior of the local belief sequences generated based on \eqref{eqn:Bayes}; we provide a proof (adapted to our notation) to keep the paper self-contained.
\begin{lemma}
Consider an agent $i\in\mathcal{S}(\theta^{\star},\theta)\cap\mathcal{R}$. Suppose $\pi_{i,0}(\theta^{\star}) > 0$. Then, the update rule \eqref{eqn:Bayes} ensures that (i) $\pi_{i,t}(\theta) \rightarrow 0$ almost surely, and (ii) $\pi_{i,\infty}(\theta^{\star})\triangleq\lim_{t\to\infty}\pi_{i,t}(\theta^{\star})$ exists almost surely, and satisfies $\pi_{i,\infty}(\theta^{\star})\geq \pi_{i,0}(\theta^{\star})$.
\label{lemma:Bayes}
\end{lemma}
\begin{proof}
Pick an agent $i\in\mathcal{S}(\theta^{\star},\theta)\cap\mathcal{R}$, and define:
\begin{equation}
\rho^{}_{i,t}(\theta)\triangleq \log\frac{\pi_{i,t}(\theta)}{\pi_{i,t}(\theta^{\star})}, \hspace{2mm} \lambda^{}_{i,t}(\theta) \triangleq \log\frac{l_i(s_{i,t}|\theta)}{l_i(s_{i,t}|\theta^{\star})}.
\end{equation}
Then, based on \eqref{eqn:Bayes}, we obtain the following recursion:
\begin{equation}
\rho_{i,t+1}(\theta)=\rho_{i,t}(\theta)+\lambda_{i,t+1}(\theta), \forall t\in\mathbb{N}.
\label{eqn:recur}
\end{equation}
Rolling out the above equation over time yields:
\begin{equation}
\rho_{i,t}(\theta)=\rho_{i,0}(\theta)+\sum \limits_{k=1}^{t}\lambda_{i,k}(\theta), \forall t\in\mathbb{N}_{+}.
\label{eqn:recursion}
\end{equation}
Notice that $\{\lambda_{i,t}(\theta)\}$ is a sequence of i.i.d. random variables with finite means and variances. In particular, it is easy to verify that each random variable $\lambda_{i,t}(\theta)$ has mean\footnote{More precisely, the mean here is obtained by using the expectation operator $\mathbb{E}^{\theta^{\star}}[\cdot]$ associated with the measure $\mathbb{P}^{\theta^{\star}}$.} given by $-D(l_i(\cdot|\theta^{\star})||l_i(\cdot|\theta))$. Thus, based on the strong law of large numbers, we have $\frac{1}{t}\sum \limits_{k=1}^{t}\lambda_{i,k}(\theta) \rightarrow -D(l_i(\cdot|\theta^{\star})||l_i(\cdot|\theta))$ almost surely. Dividing both sides of \eqref{eqn:recursion} by $t$, and taking the limit as $t$ goes to infinity, we then obtain:
\begin{equation}
\lim_{t\to\infty}\frac{1}{t}\rho_{i,t}(\theta)=-D(l_i(\cdot|\theta^{\star})||l_i(\cdot|\theta)) \hspace{1mm} \textrm{almost surely}.
\label{eqn:limit}
\end{equation}
Finally, note that based on the definition of the set $\mathcal{S}(\theta^{\star},\theta)$, $D(l_i(\cdot|\theta^{\star})||l_i(\cdot|\theta)) > 0$. It then follows from \eqref{eqn:limit} that $\rho_{i,t}(\theta) \rightarrow -\infty$ almost surely, and hence $\pi_{i,t}(\theta) \rightarrow 0$ almost surely. For any $\theta\in\Theta^{\theta^{\star}}_i$, observe that $\lambda_{i,t}(\theta)=0, \forall t\in\mathbb{N}_{+}$. It then follows from \eqref{eqn:recur} that for each $\theta\in\Theta^{\theta^{\star}}_i$, $\rho_{i,t}(\theta)=\rho_{i,0}(\theta), \forall t\in\mathbb{N}_{+}$. From the above discussion, we conclude that a limiting belief vector $\boldsymbol{\pi}_{i,\infty}$ exists almost surely, with non-zero entries corresponding to only those $\theta\in\Theta^{\theta^{\star}}_i$ for which $\pi_{i,0}(\theta) > 0$. Part (ii) of the lemma then follows readily by noting that $\pi_{i,0}(\theta^{\star}) > 0$.
\end{proof}
We are now in position to prove Theorems \ref{thm:rule1} and \ref{thm:rule2}.
\subsection{Proof of Theorem \ref{thm:rule1}}
\begin{proof} Let $\bar{\Omega}\subseteq\Omega$ denote the set of sample paths along which for each agent $i\in\mathcal{V}$, the following hold: (i) for each $\theta\in\Theta\setminus{\Theta^{\theta^{\star}}_i}$,  $\pi_{i,t}(\theta) \rightarrow 0$, and (ii) $\pi_{i,\infty}(\theta^{\star})\triangleq\lim_{t\to\infty}\pi_{i,t}(\theta^{\star})$ exists, and satisfies $\pi_{i,\infty}(\theta^{\star})\geq \pi_{i,0}(\theta^{\star})$. Recall that $\Theta^{\theta^{\star}}_i$ represents the set of hypotheses that are observationally equivalent to the true state $\theta^{\star}$ from the point of view of agent $i$. Hence, for each $\theta\in\Theta\setminus{\Theta^{\theta^{\star}}_i}$, we have $i \in \mathcal{S}(\theta^{\star},\theta)$. Based on the third condition in the statement of Theorem \ref{thm:rule1}, and Lemma \ref{lemma:Bayes}, we infer that $\bar{\Omega}$ has measure $1$. Thus, to prove the desired result, it suffices to confine our attention to the set $\bar{\Omega}$. Specifically, fix any sample path $\omega\in\bar{\Omega}$, and pick any $\epsilon > 0$. Our goal will be to establish that along the sample path $\omega$, there exists $t(\omega,\epsilon)$ such that for all $t\geq t(\omega,\epsilon)$, $\mu_{i,t}(\theta) < \epsilon$ for all $i\in\mathcal{V}$, and for all $\theta \neq \theta^{\star}$ in the dynamics given by \eqref{eqn:rule1}. This would be equivalent to establishing that the actual beliefs of all the agents on the true state can be made arbitrarily close to $1$ (since the proposed min-rule \eqref{eqn:rule1} generates a valid probability distribution over the hypothesis set at each time-step). We complete the proof in the following two steps.

\textbf{\underline{Step 1}:} \textit{Lower bounding the actual beliefs on the true state}: Consider the following scenario. During a transient phase, certain agents see private signals that cause them to temporarily lower their local beliefs on the true state. This in turn gets propagated via the min-rule \eqref{eqn:rule1} to the actual beliefs of the agents in the network. For sample paths in the set $\bar{\Omega}$, we rule out the possibility of such a transient phenomenon triggering a cascade of progressively lower beliefs on the true state. To this end, define $\gamma_1\triangleq\min_{i\in\mathcal{V}} \pi_{i,0}(\theta^{\star})$. Notice that $\gamma_1 > 0$ based on condition (iii) of the theorem. Given the choice of the sample path $\omega$, we notice that $\pi_{i,\infty}(\theta^{\star})$ exists for each $i\in\mathcal{V}$, and that $\pi_{i,\infty}(\theta^{\star}) \geq \gamma_1$. Pick a small number $\delta > 0$ such that $\delta < \gamma_1$. The following statement is then immediate. For each agent $i\in\mathcal{V}$, there exists $t_i(\omega,\delta)$, such that for all $t \geq t_i(\omega,\delta)$, $\pi_{i,t}(\theta^{\star}) \geq  \gamma_1-\delta > 0$. Define $\bar{t}_1(\omega,\delta)\triangleq\max_{i\in\mathcal{V}}t_i(\omega,\delta)$. In words, $\bar{t}_1(\omega,\delta)$ represents the time-step beyond which the local beliefs of all the agents on the true state are lower-bounded by $\gamma_1-\delta$. We ask: At such a time-step, what is the lowest actual belief held by an agent on the true state? More precisely, we define $\gamma_2(\omega)\triangleq\min_{i\in\mathcal{V}}\{\mu_{i,\bar{t}_1(\omega,\delta)}(\theta^{\star})\}$. We claim $\gamma_2(\omega) > 0$. To see this, observe that given the assumption of non-zero prior beliefs on the true state, and the structure of the proposed min-rule \eqref{eqn:rule1}, $\gamma_2(\omega)$ can be $0$ if and only if there exists some time-step $t^{'}(\omega) \leq \bar{t}_1(\omega,\delta)$ such that $\pi_{i,t^{'}(\omega)}(\theta^{\star})=0$, for some $i\in\mathcal{V}$. However, given the structure of the local Bayesian update rule \eqref{eqn:Bayes}, we would then have $\pi_{i,t}(\theta^{\star})=0$, for all $t\geq{t}^{'}(\omega)$, contradicting the fact that $\pi_{i,t}(\theta^{\star}) \geq \gamma_1-\delta > 0, \forall t \geq \bar{t}_1(\omega,\delta) \geq t^{'}(\omega), \forall i\in\mathcal{V}$ (the latter fact has already been established above). Having thus  established that $\gamma_2(\omega) > 0$, define $\eta(\omega)\triangleq\min\{\gamma_1-\delta,\gamma_2(\omega)\} > 0$. In other words, $\eta(\omega)$ lower-bounds the lowest belief (considering both local and actual beliefs) on the true state $\theta^{\star}$ held by an agent at time-step $\bar{t}_1(\omega,\delta)$. We claim the following:
\begin{equation}
\mu_{i,t}(\theta^{\star}) \geq \eta(\omega), \forall t \geq \bar{t}_1(\omega,\delta), \forall i\in\mathcal{V}.
\label{eqn:lower_bound}
\end{equation}
To see why \eqref{eqn:lower_bound} is true, fix an agent $i\in\mathcal{V}$, and consider the following chain of inequalities:
\begin{equation}
\begin{aligned}
\mu_{i,\bar{t}_1(\omega,\delta)+1}(\theta^{\star})&\overset{(a)}{=}\frac{\min\{\{\mu_{j,\bar{t}_1(\omega,\delta)}(\theta^{\star})\}_{{j\in\mathcal{N}_i}},\pi_{i,\bar{t}_1(\omega,\delta)+1}(\theta^{\star})\}}{\sum\limits_{p=1}^{m}\min\{\{\mu_{j,\bar{t}_1(\omega,\delta)}(\theta_p)\}_{{j\in\mathcal{N}_i}},\pi_{i,\bar{t}_1(\omega,\delta)+1}(\theta_p)\}}\\
&\overset{(b)}{\geq}\frac{\eta(\omega)}{\sum\limits_{p=1}^{m}\min\{\{\mu_{j,\bar{t}_1(\omega,\delta)}(\theta_p)\}_{{j\in\mathcal{N}_i}},\pi_{i,\bar{t}_1(\omega,\delta)+1}(\theta_p)\}}\\
&\overset{}{\geq}\frac{\eta(\omega)}{\sum\limits_{p=1}^{m}\pi_{i,\bar{t}_1(\omega,\delta)+1}(\theta_p)}\\
&\overset{(c)}{=}\eta(\omega),
\end{aligned}
\label{eqn:lower1}
\end{equation}
where $(a)$ is given by \eqref{eqn:rule1}, $(b)$ follows from the way $\eta(\omega)$ is defined and by noting that $\pi_{i,t}(\theta^{\star}) \geq \eta(\omega), \forall t\geq \bar{t}_1(\omega,\delta), \forall i\in\mathcal{V}$, and $(c)$ follows by noting that the local belief vectors generated via \eqref{eqn:Bayes} (at each time-step) are valid probability distributions over the hypothesis set $\Theta$, and hence $\sum\limits_{p=1}^{m}\pi_{i,\bar{t}_1(\omega,\delta)+1}(\theta_p)=1$. Since the above reasoning applies to every agent in the network, we can keep repeating it to establish \eqref{eqn:lower_bound} via induction. 

\textbf{\underline{Step 2}:} \textit{Upper bounding the actual beliefs on each false hypothesis}:  The key observation that guides the rest of the proof is as follows. While Step 1 of the proof ensures that the beliefs (both local and actual) of each agent on the true state $\theta^{\star}$ are lower-bounded by $\eta(\omega)>0$ after a finite period of time (given by $\bar{t}_1(\omega,\delta)$), Lemma \ref{lemma:Bayes} guarantees that the local beliefs on any false hypothesis $\theta$ will eventually become arbitrarily small (and in particular, smaller than $\eta(\omega)$) for each agent $i\in\mathcal{S}(\theta^{\star},\theta)$, on the sample path $\omega\in\bar{\Omega}$ under consideration. In what follows, we investigate how this impacts the actual beliefs of the agents in the network. To this end, given an $\epsilon > 0$, pick a small $\bar{\epsilon}(\omega) > 0$ such that  $\bar{\epsilon}(\omega) < \min\{\eta(\omega),\epsilon\}$. Fix a hypothesis $\theta \neq \theta^{\star}$. By virtue of condition (i) of the theorem, we know that $|\mathcal{S}(\theta^{\star},\theta)| > 0$. Let $q=d(\mathcal{G})+2$, where $d(\mathcal{G})$ represents the diameter of the graph $\mathcal{G}$. Then, based on Lemma \ref{lemma:Bayes}, for each $i\in\mathcal{S}(\theta^{\star},\theta)$, there exists $t^{\theta}_i(\omega,\bar{\epsilon}(\omega))$ such that for all $t\geq t^{\theta}_i(\omega,\bar{\epsilon}(\omega))$, $\pi_{i,t}(\theta) \leq \bar{\epsilon}^{q}(\omega)$. Define 
\begin{equation}
\bar{t}^{\theta}_2(\omega,\delta,\bar{\epsilon}(\omega))\triangleq\max\{\bar{t}_1(\omega,\delta), \max_{i\in\mathcal{S}(\theta^{\star},\theta)}\{t^{\theta}_i(\omega,\bar{\epsilon}(\omega))\}\}.
\label{eqn:crit_time}
\end{equation}
Throughout the rest of the proof, we suppress the dependence of $\bar{t}_2$ on $\theta,\omega,\delta$ and $\bar{\epsilon}(\omega)$ to avoid cluttering the exposition.
For any agent $i\in\mathcal{S}(\theta^{\star},\theta)$, we obtain the following chain of inequalities:
\begin{equation}
\begin{aligned}
\mu_{i,\bar{t}_2+1}(\theta)&\overset{(a)}{=}\frac{\min\{\{\mu_{j,\bar{t}_2}(\theta)\}_{{j\in\mathcal{N}_i}},\pi_{i,\bar{t}_2+1}(\theta)\}}{\sum\limits_{p=1}^{m}\min\{\{\mu_{j,\bar{t}_2}(\theta_p)\}_{{j\in\mathcal{N}_i}},\pi_{i,\bar{t}_2+1}(\theta_p)\}}\\
&\overset{(b)}{\leq}\frac{\bar{\epsilon}^{q}(\omega)}{\sum\limits_{p=1}^{m}\min\{\{\mu_{j,\bar{t}_2}(\theta_p)\}_{{j\in\mathcal{N}_i}},\pi_{i,\bar{t}_2+1}(\theta_p)\}}\\
\\
&\overset{}{\leq}\frac{\bar{\epsilon}^{q}(\omega)}{{\min\{\{\mu_{j,\bar{t}_2}(\theta^{\star})\}_{{j\in\mathcal{N}_i}},\pi_{i,\bar{t}_2+1}(\theta^{\star})\}}}\\
&\overset{(c)}{\leq}\frac{\bar{\epsilon}^{q}(\omega)}{\eta(\omega)}\\
&\overset{(d)}{<}\bar{\epsilon}^{(q-1)}(\omega)\leq\bar{\epsilon}(\omega)<\epsilon,
\end{aligned}
\label{eqn:upper_bound}
\end{equation}
where $(a)$ is given by \eqref{eqn:rule1}, $(b)$ follows from the fact that for each $i\in\mathcal{S}(\theta^{\star},\theta)$, we have $\pi_{i,t}(\theta) \leq \bar{\epsilon}^{q}(\omega), \forall t \geq \bar{t}_2$, $(c)$ follows from \eqref{eqn:lower_bound} and \eqref{eqn:crit_time}, and $(d)$ follows from the way $\bar{\epsilon}$ has been chosen. In particular, note that the above chain of reasoning used to arrive at \eqref{eqn:upper_bound} applies to subsequent time-steps as well. We thus conclude:
\begin{equation}
\mu_{i,t}(\theta) < \bar{\epsilon}^{(q-1)}(\omega), \forall t \geq \bar{t}_2+1, \forall i \in \mathcal{S}(\theta^{\star},\theta).
\label{eqn:upp_bound_source}
\end{equation}
We now wish to investigate how the effect of \eqref{eqn:upp_bound_source} propagates through the rest of the network. If $\mathcal{V}\setminus\mathcal{S}(\theta^{\star},\theta)$ is empty, then we have reached the desired conclusion w.r.t. the false hypothesis $\theta$. If not, define
\begin{equation}
\mathcal{L}^{(\theta^{\star},\theta)}_1\triangleq \{i\in\{\mathcal{V}\setminus\mathcal{S}(\theta^{\star},\theta)\} \hspace{1mm} {:} \hspace{1mm} |\mathcal{N}_i\cap\mathcal{S}(\theta^{\star},\theta)| > 0\}
\label{eqn:level1}
\end{equation}
as the set of immediate out-neighbors of the source set $\mathcal{S}(\theta^{\star},\theta)$. By virtue of condition (ii) of the theorem, if $\mathcal{V}\setminus\mathcal{S}(\theta^{\star},\theta)$ is non-empty, then $\mathcal{L}^{(\theta^{\star},\theta)}_1$ as defined above is also non-empty. Consider any agent $i \in \mathcal{L}^{(\theta^{\star},\theta)}_1$. By definition, agent $i$ has a neighbor in $\mathcal{S}(\theta^{\star},\theta)$ satisfying \eqref{eqn:upp_bound_source}. This observation coupled with equations \eqref{eqn:lower_bound}, \eqref{eqn:crit_time} can be used to obtain a similar chain of inequalities as the ones featuring in \eqref{eqn:upper_bound}. Specifically, we obtain:
\begin{equation}
\mu_{i,t}(\theta) < \bar{\epsilon}^{(q-2)}(\omega), \forall t \geq \bar{t}_2+2, \forall i \in \mathcal{L}^{(\theta^{\star},\theta)}_1.
\label{eqn:upp_bound_l1}
\end{equation}
With $\mathcal{L}^{(\theta^{\star},\theta)}_0\triangleq\mathcal{S}(\theta^{\star},\theta)$, the above arguments can be repeated by successively defining the sets $\mathcal{L}^{(\theta^{\star},\theta)}_r, 1\leq r \leq d(\mathcal{G})$ as follows:
\begin{equation}
\mathcal{L}^{(\theta^{\star},\theta)}_r\triangleq \{i\in\mathcal{V}\setminus\{\bigcup_{c=0}^{r-1}\mathcal{L}^{(\theta^{\star},\theta)}_c\} \hspace{1mm} {:} \hspace{1mm}  |\mathcal{N}_i\cap \{\bigcup_{c=0}^{r-1}\mathcal{L}^{(\theta^{\star},\theta)}_c\}| > 0\}.
\label{eqn:levelr}
\end{equation}
Whenever $\mathcal{V}\setminus\{\bigcup_{c=0}^{r-1}\mathcal{L}^{(\theta^{\star},\theta)}_c\}$ is non-empty, condition (ii) of the theorem implies that $\mathcal{L}^{(\theta^{\star},\theta)}_r$ will also be non-empty. 
One can then easily verify via induction on $r$ that:
\begin{equation}
\mu_{i,t}(\theta) < \bar{\epsilon}^{(q-(r+1))}(\omega), \forall t \geq \bar{t}_2+(r+1), \forall i \in \mathcal{L}^{(\theta^{\star},\theta)}_r,
\label{eqn:upp_bound_lr}
\end{equation}
where $1\leq r \leq d(\mathcal{G})$. Noting that $q=d(\mathcal{G})+2$, we obtain the desired result that $\mu_{i,t}(\theta) < \bar{\epsilon}(\omega) < \epsilon$, $\forall t\geq \bar{t}_2+d(\mathcal{G})+1, \forall i \in \mathcal{V}$. An identical argument as the one presented above can be made for each false hypothesis $\theta\neq\theta^{\star}$. This completes the proof.
\end{proof}
\subsection{Proof of Theorem \ref{thm:rule2}}
\label{sec:proofthm2}
\begin{proof}
Consider an $f$-local adversarial set $\mathcal{A}\subset\mathcal{V}$, and let $\mathcal{R}=\mathcal{V}\setminus\mathcal{A}$. We study two separate cases. 

\underline{\textbf{Case 1:}} Consider a regular agent $i\in\mathcal{R}$ such that $|\mathcal{N}_i| < (2f+1)$. Based on  condition (i) of the theorem, we claim that $i\in\mathcal{S}(\theta_p,\theta_q)$, for every pair $\theta_p,\theta_q \in \Theta$. We prove this claim via contradiction. To do so, suppose there exists a pair $\theta_p,\theta_q\in\Theta$, such that $i\in\mathcal{V}\setminus\mathcal{S}(\theta_p,\theta_q)$. As $|\mathcal{N}_i| < (2f+1)$, the set $\{i\}$ is clearly not $(2f+1)$-reachable (see Def. \ref{defn:rreachable}). Thus, $\mathcal{G}$ is not strongly $(2f+1)$-robust w.r.t. the source set $\mathcal{S}(\theta_p,\theta_q)$, a fact that contradicts condition (i) of the theorem. Thus, we have established that for networks satisfying condition (i) of the theorem, regular agents with fewer than $(2f+1)$ neighbors can distinguish between every pair of hypotheses. Lemma \ref{lemma:Bayes} then implies that such agents can discern the true state $\theta^{\star}$ by simply running the local Bayesian estimator \eqref{eqn:Bayes}, and updating actual beliefs via \eqref{eqn:rule3}.

\underline{\textbf{Case 2:}} We now focus only on regular agents $i$ satisfying $|\mathcal{N}_i| \geq (2f+1)$. For this case, the structure of the proof mirrors that of Theorem \ref{thm:rule1}; we thus only elaborate on details that are specific to tackling the aspect of adversarial agents. A key property of the proposed LFRHE algorithm that will be used throughout the proof is as follows. For any $i\in\mathcal{R}$, and any $\theta\in\Theta$, the filtering operation of the LFRHE algorithm ensures that at each time-step  $t\in\mathbb{N}$, we have:
\begin{equation}
\mu_{j,t}(\theta) \in Conv(\Psi^{\theta}_{i,t}), \forall j \in \mathcal{M}^{\theta}_{i,t},
\label{eqn:LHRHE_property}
\end{equation}
where 
\begin{equation}
\Psi^{\theta}_{i,t} \triangleq \{\mu_{j,t}(\theta) \hspace{1mm} {:} \hspace{1mm}  j\in\mathcal{N}_i\cap\mathcal{R}\},
\label{eqn:set}
\end{equation}
and $Conv(\Psi^{\theta}_{i,t})$ is used to denote the convex hull formed by the points in the set $\Psi^{\theta}_{i,t}$. In other words, any neighboring belief (on a particular hypothesis) that agent $i$ uses in the update rule \eqref{eqn:rule2} lies in the convex hull of the actual beliefs of its regular neighbors (on that particular hypothesis). To see why \eqref{eqn:LHRHE_property} is true, partition the neighbor set $\mathcal{N}_i$ of a regular agent into three sets $\mathcal{U}^{\theta}_{i,t}, \mathcal{M}^{\theta}_{i,t}$, and $\mathcal{J}^{\theta}_{i,t}$ as follows. Sets $\mathcal{U}^{\theta}_{i,t}$ and $\mathcal{J}^{\theta}_{i,t}$ are each of cardinality $f$, and contain neighbors of agent $i$ that transmit the highest $f$ and the lowest $f$ actual beliefs respectively, on the hypothesis $\theta$, to agent $i$ at time-step $t$. The set $\mathcal{M}^{\theta}_{i,t}$ contains the remaining neighbors of agent $i$, and is non-empty at every time-step since $|\mathcal{N}_i| \geq (2f+1)$. If $\mathcal{M}^{\theta}_{i,t}\cap\mathcal{A}=\emptyset$, then \eqref{eqn:LHRHE_property} holds trivially. Thus, consider the case when there are adversaries in the set $\mathcal{M}^{\theta}_{i,t}$, i.e., $\mathcal{M}^{\theta}_{i,t}\cap\mathcal{A} \neq \emptyset$. Given the $f$-locality of the adversarial model, and the nature of the filtering operation in the LFRHE algorithm, we infer that for each $j\in\mathcal{M}^{\theta}_{i,t}\cap\mathcal{A}$, there exist regular agents $u,v\in\mathcal{N}_i\cap\mathcal{R}$, such that $u\in\mathcal{U}^{\theta}_{i,t}$, $v\in\mathcal{J}^{\theta}_{i,t}$, and $\mu_{v,t}(\theta) \leq \mu_{j,t}(\theta) \leq \mu_{u,t}(\theta)$. This establishes our claim regarding equation \eqref{eqn:LHRHE_property}.

With the above property in hand, our goal will be to now establish each of the two steps in the proof of Theorem \ref{thm:rule1}. To this end, let $\bar{\Omega}\subseteq\Omega$ denote the set of sample paths along which for each agent $i\in\mathcal{R}$, the following hold: (i) for each $\theta\in\Theta\setminus{\Theta^{\theta^{\star}}_i}$,  $\pi_{i,t}(\theta) \rightarrow 0$, and (ii) $\pi_{i,\infty}(\theta^{\star})\triangleq\lim_{t\to\infty}\pi_{i,t}(\theta^{\star})$ exists, and satisfies $\pi_{i,\infty}(\theta^{\star})\geq \pi_{i,0}(\theta^{\star})$. Based on condition (ii)  of the theorem, and Lemma \ref{lemma:Bayes}, we infer that $\bar{\Omega}$ has measure $1$. Thus, as in Theorem \ref{thm:rule1}, fix a sample path $\omega\in\bar{\Omega}$, and pick $\epsilon > 0$. Define $\gamma_1=\min_{i\in\mathcal{R}} \pi_{i,0}(\theta^{\star})$, pick a small number $\delta > 0$ satisfying $\delta < \gamma_1$, and observe that for each agent $i\in\mathcal{R}$, there exists $t_i(\omega,\delta)$, such that for all $t \geq t_i(\omega,\delta)$, $\pi_{i,t}(\theta^{\star}) \geq  \gamma_1-\delta > 0$. Define $\bar{t}_1(\omega,\delta)\triangleq\max_{i\in\mathcal{R}}t_i(\omega,\delta)$ and $\gamma_2(\omega)\triangleq\min_{i\in\mathcal{R}}\{\mu_{i,\bar{t}_1(\omega,\delta)}(\theta^{\star})\}$. As before, we claim $\gamma_2(\omega) > 0$. To establish this claim, we need to answer the following question: Can an adversarial agent cause its out-neighbors to set their actual beliefs on $\theta^{\star}$ to be $0$ by setting its own actual belief on $\theta^{\star}$ to be $0$? We argue that this is impossible under the LFRHE algorithm. By way of contradiction, suppose there exists a time-step ${t}'(\omega)$ satisfying:
\begin{equation}
{t}'(\omega)=\min\{t\in\mathbb{N}\hspace{1mm}{:} \hspace{1mm} \exists i \in \mathcal{R} \hspace{1mm}\textrm{with} \hspace{1mm}\mu_{i,t}(\theta^{\star})=0\}.\
\end{equation}
In words, $t'(\omega)$ represents the first time-step when some regular agent $i$ sets its actual belief on the true hypothesis to be zero. Clearly, $t'(\omega)\neq 0$ based on condition (ii) of the theorem. Suppose  ${t}'(\omega)$ is some positive integer, and focus on how agent $i$ updates $\mu_{i,{t}'(\omega)}(\theta^{\star})$ based on \eqref{eqn:rule2}. Following similar arguments as in the proof of Theorem \ref{thm:rule1}, we know that $\pi_{i,t}(\theta^{\star}) > 0, \forall t\in \mathbb{N}, \forall i \in \mathcal{R}.$ At the same time, every belief featuring in the set $\Psi^{\theta^{\star}}_{i,{t}'(\omega)-1}$ (as defined in equation \eqref{eqn:set}) is strictly positive based on the way ${t}'(\omega)$ is defined. In light of the above arguments, and based on \eqref{eqn:LHRHE_property}, \eqref{eqn:set}, we infer:
\begin{equation}
\min\{\{\mu_{j,{t}'(\omega)-1}(\theta^{\star})\}_{j\in\mathcal{M}^{\theta^{\star}}_{i,{t}'(\omega)-1}},\pi_{i,{t}'(\omega)}(\theta^{\star})\} > 0.
\end{equation}
Thus, based on \eqref{eqn:rule2}, we must have $\mu_{i,{t}'(\omega)}(\theta^{\star}) > 0$, yielding the desired contradiction. With $\eta(\omega)\triangleq\min\{\gamma_1-\delta,\gamma_2(\omega)\} > 0$, one can easily  verify the following:
\begin{equation}
\mu_{i,t}(\theta^{\star}) \geq \eta(\omega), \forall t \geq \bar{t}_1(\omega,\delta), \forall i\in\mathcal{R}.
\label{eqn:lower2}
\end{equation}
In particular, \eqref{eqn:lower2} follows by (i) noting that for each $i \in \mathcal{R}$, $\pi_{i,\bar{t}_1(\omega,\delta)+1}(\theta^{\star}) \geq \eta(\omega)$, and each belief featuring in the set $\Psi^{\theta^{\star}}_{i,\bar{t}_1(\omega,\delta)}$ is lower bounded by $\eta(\omega)$, (ii) leveraging \eqref{eqn:LHRHE_property}, \eqref{eqn:set}, and (iii) using a similar string of arguments as those used to arrive at \eqref{eqn:lower1}. This completes Step 1.

To proceed with Step 2 (i.e., upper-bounding the actual beliefs on each false hypothesis), given an $\epsilon > 0$, pick a small $\bar{\epsilon}(\omega) > 0$ such that  $\bar{\epsilon}(\omega) < \min\{\eta(\omega),\epsilon\}$. Fix a hypothesis $\theta \neq \theta^{\star}$, let $q=n+1$, and note that based on Lemma \ref{lemma:Bayes}, for each $i\in\mathcal{S}(\theta^{\star},\theta)\cap\mathcal{R}$, there exists $t^{\theta}_i(\omega,\bar{\epsilon}(\omega))$ such that for all $t\geq t^{\theta}_i(\omega,\bar{\epsilon}(\omega))$, $\pi_{i,t}(\theta) \leq \bar{\epsilon}^{q}(\omega)$. Define 
\begin{equation}
\bar{t}_2\triangleq\max\{\bar{t}_1(\omega,\delta), \max_{i\in\mathcal{S}(\theta^{\star},\theta)\cap\mathcal{R}}\{t^{\theta}_i(\omega,\bar{\epsilon}(\omega))\}\},
\label{eqn:crit_time2}
\end{equation}
where we have suppressed the dependence of $\bar{t}_2$ on $\theta,\omega,\delta$ and $\bar{\epsilon}(\omega)$ as in the proof of Theorem \ref{thm:rule1}. For any agent $i\in\mathcal{S}(\theta^{\star},\theta)\cap\mathcal{R}$, observe that
\begin{equation}
\min\{\{\mu_{j,\bar{t}_2}(\theta^{\star})\}_{j\in\mathcal{M}^{\theta^{\star}}_{i,\bar{t}_2}},\pi_{i,\bar{t}_2+1}(\theta^{\star})\} \geq \eta(\omega).
\end{equation}
Combining the above with a similar line of argument as used to arrive at \eqref{eqn:upper_bound}, we obtain:
\begin{equation}
\mu_{i,t}(\theta) < \bar{\epsilon}^{(q-1)}(\omega), \forall t \geq \bar{t}_2+1, \forall i \in \mathcal{S}(\theta^{\star},\theta)\cap\mathcal{R}.
\label{eqn:upp_bound_source2}
\end{equation}
If $\mathcal{V}\setminus\mathcal{S}(\theta^{\star},\theta)$ is empty, then we are done. Else, define 
\begin{equation}
\mathcal{L}^{(\theta^{\star},\theta)}_1\triangleq \{i\in\{\mathcal{V}\setminus\mathcal{S}(\theta^{\star},\theta)\} \hspace{1mm} {:} \hspace{1mm} |\mathcal{N}_i\cap\mathcal{S}(\theta^{\star},\theta)| \geq (2f+1)\}.
\label{eqn:level1thm2}
\end{equation}
Whenever $\mathcal{V}\setminus\mathcal{S}(\theta^{\star},\theta)$ is non-empty, we claim that $\mathcal{L}^{(\theta^{\star},\theta)}_1$ (as defined above) is also non-empty based on condition (i) of the theorem. To see this, note that if $\mathcal{L}^{(\theta^{\star},\theta)}_1$ is empty, then $\mathcal{C}=\mathcal{V}\setminus\mathcal{S}(\theta^{\star},\theta)$ is not $(2f+1)$-reachable, violating the fact that $\mathcal{G}$ is strongly $(2f+1)$-robust w.r.t. $\mathcal{S}(\theta^{\star},\theta)$. We claim
\begin{equation}
\min_{j\in\mathcal{M}^{\theta}_{i,\bar{t}_2+1}}{\mu_{j,\bar{t}_2+1}(\theta)} < \bar{\epsilon}^{(q-1)}(\omega), \forall i\in\mathcal{L}^{(\theta^{\star},\theta)}_1\cap\mathcal{R}.
\label{eqn:bound1}
\end{equation}
To verify the above claim, pick any agent $i\in\mathcal{L}^{(\theta^{\star},\theta)}_1\cap\mathcal{R}$. When $|\mathcal{M}^{\theta}_{i,\bar{t}_2+1}\cap\{\mathcal{S}(\theta^{\star},\theta)\cap\mathcal{R}\}|>0$, the claim follows immediately based on \eqref{eqn:upp_bound_source2}. Consider the case when $|\mathcal{M}^{\theta}_{i,\bar{t}_2+1}\cap\{\mathcal{S}(\theta^{\star},\theta)\cap\mathcal{R}\}|=0$. Since $i\in\mathcal{L}^{(\theta^{\star},\theta)}_1$, it has at least $(2f+1)$ neighbors in $\mathcal{S}(\theta^{\star},\theta)$, out of which at least $f+1$ are regular based on the $f$-locality of the adversarial model. Since the set $\mathcal{J}^{\theta}_{i,\bar{t}_2+1}$ has cardinality $f$, it must then be that $|\mathcal{U}^{\theta}_{i,\bar{t}_2+1}\cap\{\mathcal{S}(\theta^{\star},\theta)\cap\mathcal{R}\}| > 0$. Let $u\in\mathcal{U}^{\theta}_{i,\bar{t}_2+1}\cap\{\mathcal{S}(\theta^{\star},\theta)\cap\mathcal{R}\}$. Based on the way $\mathcal{M}^{\theta}_{i,\bar{t}_2+1}$ is defined, it must  be that $\mu_{j,\bar{t}_2+1}(\theta) \leq \mu_{u,\bar{t}_2+1}(\theta) < \bar{\epsilon}^{(q-1)}(\omega), \forall j \in \mathcal{M}^{\theta}_{i,\bar{t}_2+1}$, where the last inequality follows from \eqref{eqn:upp_bound_source2}. This establishes our claim regarding \eqref{eqn:bound1}. Consider the update of $\mu_{i,\bar{t}_2+2}(\theta)$ based on \eqref{eqn:rule2}. In light of the above arguments (that apply identically to subsequent time-steps as well), the numerator of the fraction on the RHS of \eqref{eqn:rule2} is upper-bounded by $\bar{\epsilon}^{(q-1)}(\omega)$, while the denominator is lower-bounded by $\eta(\omega)$. This leads to the following conclusion:
\begin{equation}
\mu_{i,t}(\theta) < \bar{\epsilon}^{(q-2)}(\omega), \forall t \geq \bar{t}_2+2, \forall i \in \mathcal{L}^{(\theta^{\star},\theta)}_1\cap\mathcal{R}.
\label{eqn:upp_bound_l1thm2}
\end{equation}
With $\mathcal{L}^{(\theta^{\star},\theta)}_0\triangleq\mathcal{S}(\theta^{\star},\theta)$, we recursively define the sets $\mathcal{L}^{(\theta^{\star},\theta)}_r, 1\leq r \leq (n-1)$ as follows:
\begin{equation}
\mathcal{L}^{(\theta^{\star},\theta)}_r\triangleq \{i\in\mathcal{V}\setminus\{\bigcup_{c=0}^{r-1}\mathcal{L}^{(\theta^{\star},\theta)}_c\} \hspace{1mm} {:} \hspace{1mm}  |\mathcal{N}_i\cap \{\bigcup_{c=0}^{r-1}\mathcal{L}^{(\theta^{\star},\theta)}_c\}| \geq (2f+1)\}.
\end{equation}
We complete the proof by inducting on $r$. To this end, suppose the following holds for all $0\leq r \leq (n-2)$:
\begin{equation}
\mu_{i,t}(\theta) < \bar{\epsilon}^{(q-(r+1))}(\omega), \forall t \geq \bar{t}_2+(r+1), \forall i \in \mathcal{L}^{(\theta^{\star},\theta)}_r\cap\mathcal{R}.
\label{eqn:upp_bound_lrthm2}
\end{equation}
The claim extends to the case when $r=(n-1)$ by noting that (i)  $\mathcal{L}^{(\theta^{\star},\theta)}_{(n-1)}$ is non-empty if $\mathcal{V}\setminus\{\bigcup_{c=0}^{(n-2)}\mathcal{L}^{(\theta^{\star},\theta)}_c\}$ is non-empty (based on condition (i) of the theorem), (ii) any agent $i\in\mathcal{L}^{(\theta^{\star},\theta)}_{(n-1)}\cap\mathcal{R}$ has at least $(2f+1)$ neighbors in the set $\bigcup_{c=0}^{(n-2)}\mathcal{L}^{(\theta^{\star},\theta)}_c$, of which at least $f+1$ are regular (based on the $f$-locality of the adversarial model), and (iii) using the induction hypothesis and arguments similar to those used for arriving at \eqref{eqn:upp_bound_l1thm2}. Finally, note that the sets $\mathcal{L}^{(\theta^{\star},\theta)}_r$  are constructed in a way such that all agents in $\mathcal{R}$ are covered. The rest of the proof is identical to that of Theorem \ref{thm:rule1}.
\end{proof}
\section{Conclusion}
In this paper, we introduced a distributed learning rule that differs fundamentally from those existing in the literature, in the sense, that it does not rely on any consensus-based belief aggregation protocol. Using a novel sample path based analysis technique, we established its consistency under minimal requirements on the information structures of the agents and the communication graph. We then showed that a significant benefit of the proposed learning rule is that it can be easily and efficiently modified to account for the presence of misbehaving agents in the network, modeled via the Byzantine adversary model.  
Ongoing work involves performing a detailed convergence rate analysis to see how such rates compare with those existing in literature. Extensions to time-varying graphs are also of interest.
\bibliographystyle{unsrt}
\bibliography{refs}
\end{document}